\documentclass[letterpaper,10pt, conference]{ieeeconf}
\IEEEoverridecommandlockouts
\usepackage{amsmath,amssymb,amsthm,epsfig,color,subfigure,empheq,graphicx}
\usepackage[algo2e]{algorithm2e}
\usepackage{enumerate,url,wasysym,balance,enumerate,dsfont}
\usepackage{algorithm,algpseudocode,stackengine}	
\usepackage{arydshln}
\usepackage{tikz,balance}
\usepackage[english]{babel}
\usepackage{xfrac}
\usetikzlibrary{matrix,decorations.pathreplacing}
\usepackage{tikz}
\newcommand{\wye}{\mathbin{\tikz[x=1ex,y=1ex]{\draw[line width=.1ex] (0,0)--(45:1)--++(-45:1) (45:1)--++(0,1);}}}

\newcommand \bef{\mathbf{f}}

\newcommand \bi{\mathbf{i}}

\newcommand \br{\mathbf{r}}

\newcommand \bv{\mathbf{v}}

\newcommand \bx{\mathbf{x}}

\newcommand \bA{\mathbf{A}}
\newcommand \bB{\mathbf{B}}

\newcommand \bone{\mathbf{1}}

\newcommand \bD{\mathbf{D}}
\newcommand \bP{\mathbf{P}}

\newcommand \bL{\mathbf{L}}

\newcommand \bV{\mathbf{V}}

\newcommand \bY{\mathbf{Y}}

\newcommand \bR{\mathbf{R}}
\newcommand \bW{\mathbf{W}}

\newcommand \mcG{\mathcal{G}}

\newcommand \mcN{\mathcal{N}}

\newcommand \mcE{\mathcal{E}}

\newcommand \bzero{\boldsymbol{0}}

\newcommand \whf{\widehat{\mathbf{f}}}
\newcommand \whL{\widehat{\mathbf{L}}}
\newcommand \whR{\widehat{\mathbf{R}}}
\newcommand \whB{\widehat{\mathbf{B}}}
\newcommand \wtB{\widetilde{\mathbf{B}}}
\newcommand \wtP{\widetilde{\mathbf{P}}}

\DeclareMathOperator{\range}{range}

\DeclareMathOperator{\nullspace}{null}

\DeclareMathOperator{\diag}{diag}

\DeclareMathOperator{\dimension}{dim}

\newtheorem{proposition}{Proposition}
\newtheorem{lemma}{Lemma}

\newtheorem{theorem}{Theorem}

\newtheorem{remark}{Remark}
\newtheorem{assumption}{Assumption}

\allowdisplaybreaks

\newcommand\groupequation[2][17pt]{%
  \setbox0=\hbox{$\displaystyle#2$}%
  \stackengine{0pt}{\copy0}{%
    \makebox[\linewidth]{\hfill$\left.\rule{0pt}{\ht0}\right\}$\kern#1}}
    {O}{c}{F}{T}{L}
}
  
\title{Time-domain Generalization of Kron Reduction}

\author{Manish K. Singh, Sairaj Dhople, Florian D\"orfler, and Georgios B. Giannakis
\thanks{
M.~K. Singh, S.~Dhople, and G.~B.~Giannakis are with the Department of Electrical \& Computer Engineering, University of Minnesota, Minneapolis, MN USA (e-mail:
\{msingh,~sdhople,~georgios\}@umn.edu).
F.~D\"orfler is with the Department of Information Technology and Electrical Engineering, ETH Z\"urich, Switzerland (e-mail: doerfler@control.ee.ethz.ch). 
}
}
\begin{document}
\maketitle
\thispagestyle{empty}
\allowdisplaybreaks
\begin{abstract}
Kron reduction is a network-reduction method that eliminates nodes with zero current injections from electrical networks operating in sinusoidal steady state. In the time domain, the state-of-the-art application of Kron reduction has been in networks with transmission lines that have constant $R/L$ ratios. This paper considers $RL$ networks without such restriction and puts forth a provably exact time-domain generalization of Kron reduction. Exemplifying empirical tests on a $\wye-\Delta$ network are provided to validate the analytical results. 
\end{abstract}

\section{Introduction}
Complex electrical networks are encountered in several engineering domains from integrated circuits to power grids. Oftentimes, a subset of nodes in such networks feature no actuation or sensing; henceforth referred to as \emph{interior} nodes. To facilitate analysis and computation, it is desirable to eliminate interior nodes and obtain reduced network models that exclusively retain the extant \emph{boundary} nodes. The workhorse enabling reduction of electrical networks derives from the classical \emph{Kron reduction}~\cite{Kron}. A familiar example of this is the elemental wye-delta ($\wye$-$\Delta$) transform. (See Fig.~\ref{fig:Overview}.) The method is well defined when all excitations are in sinusoidal steady state and all network interconnections are modeled as impedances at a fixed frequency. Kron reduction in such a setting boils down to computing a Schur complement of the admittance matrix (that establishes the algebraic map between nodal voltages and current injections). The effort~\cite{Florian-2013-Kron} provides a comprehensive survey of Kron reduction and establishes connections to a wide range of graph- and system-theoretic constructs; similarly~\cite{Sairaj14synchronization,Monshizadeh-2018} highlight some recent extensions for diverse applications. 

Interestingly, Kron reduction is not as widely studied in the time domain with arbitrary excitation, wherein the governing dynamics are differential-algebraic equations (DAEs). Exact model-reduction results in these settings are restricted to homogeneous networks, which assume lines to have constant $R/L$ ratios, and include purely resistive and inductive networks as special cases~\cite{Caliskan12CDC, Caliskan14Automatica}. Attempts to address generalized settings yield limited accuracy guarantees~\cite{Floriduz19Kron}.

Pursuing a time-domain generalization of Kron reduction, this work considers an $RL$ network without the restrictive constant $R/L$ constraint. For the considered setting, this work contributes a projection-based generalized time-domain reduced model with two prominent advantages: i)~the reduction is exact, implying equivalence to the full-order model; and ii)~the time-domain analysis permits inclusion of arbitrary initial conditions (as opposed to frequency-domain approaches). Phasor-domain Kron reduction and time-domain Kron reduction with constant $R/L$ ratios are recovered as special cases.  
\begin{figure}[t!]
\centering 
\includegraphics[scale=0.8]{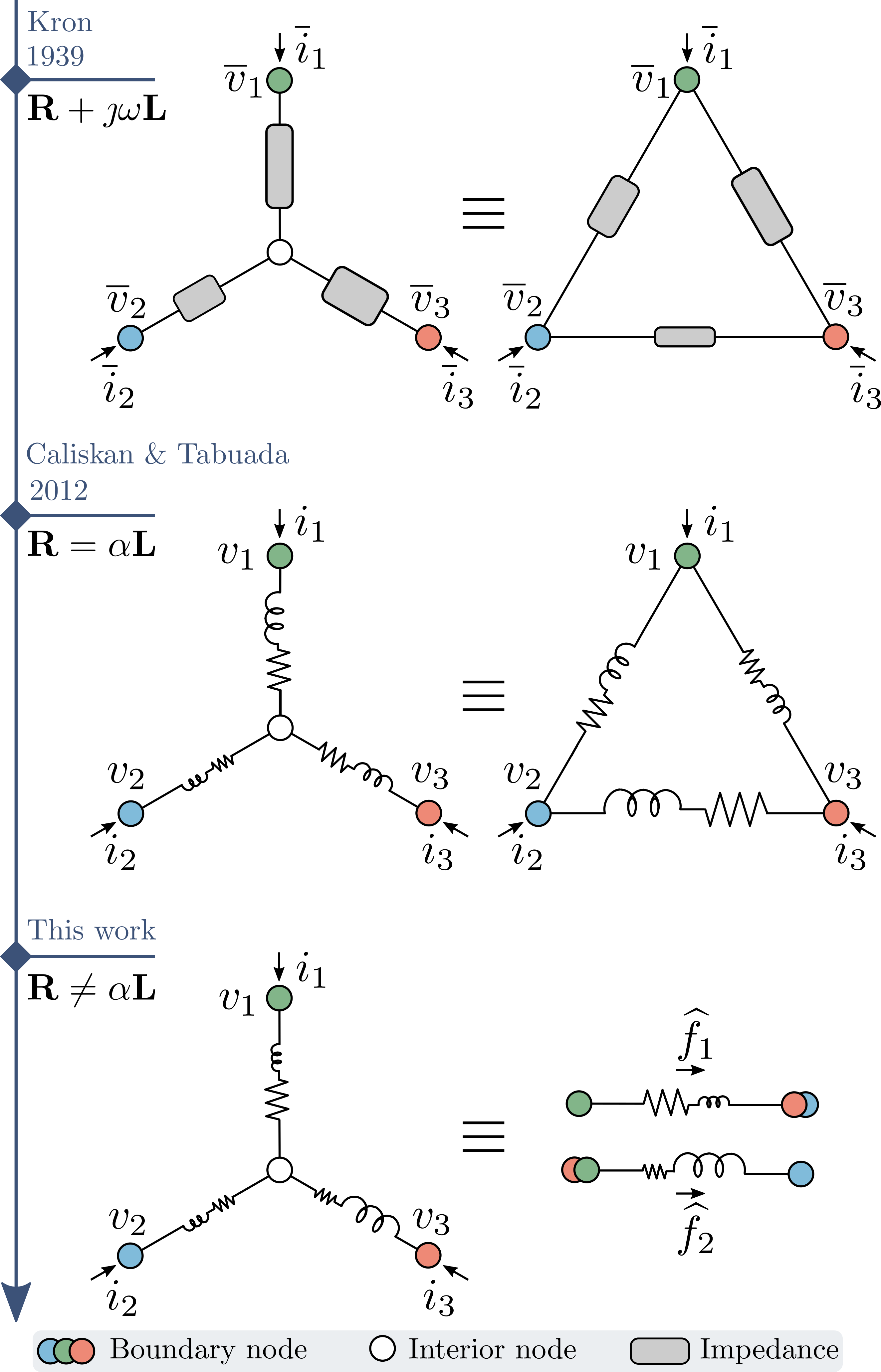}
\caption{Prominent existing results and proposed advancement towards electrical network reduction. State-of-the art for phasor domain is the classical Kron reduction (top); that in the time domain is restricted to networks with constant $R/L$ ratios (middle). We provide a generalization in the time domain (bottom) that recovers prior results as special cases.}\label{fig:Overview}
\vspace{-1em}
\end{figure}
\section{Preliminaries}\label{prelim}
\subsection{Phasor Representation} \label{sec:Notation}
In sinusoidal steady state at frequency $\omega$, we express time-domain signals as $x(t)=|x|\cos(\omega t+\theta_x)$, 
where $(|x|,\theta_x)$ are constants. For analytical ease, we represent $x(t)$ by a corresponding complex-valued rotating vector
\begin{equation*}
     \overrightarrow{x}(t)=x(t)+\jmath x(t-\sfrac{\pi}{2\omega})=|x|\mathrm{e}^{\jmath\theta_x}\mathrm{e}^{\jmath \omega t} = \overline{x} \mathrm{e}^{\jmath \omega t}.
\end{equation*}
The complex constant quantity $\overline{x} = |x|\mathrm{e}^{\jmath\theta_x}$ is referred to as a phasor. Dynamics satisfied by real-valued signals $x(t)$ in a linear and time-invariant system are also satisfied by $\overrightarrow{x}(t)$. This facilitates translating differential equations in $x(t)$ to algebraic equations in $\overline{x}$ and underscores the popularity of phasors in steady-state analysis of electrical networks. 

\subsection{Electrical-network Model}
Consider a single-phase $RL$ network described by a connected graph $\mcG=(\mcN, \mcE)$. The node set is indexed as $\mcN=\{1,\dots,N\}$, and let $E=|\mcE|$. Arbitrarily assigning directions, an edge $e$ from node $m$ to $n$ is denoted as $e=(m,n)$. Line resistances and inductances for an edge $e \in \mcE$ are denoted by $r_e,~\ell_e\geq0$. The topology of $\mcG$ is captured by the incidence matrix $\bB\in\{0,\pm 1\}^{N\times E}$ with entries $B_{k,e}=1(-1)$ if $k=m(n)$ when $\exists~e=(m,n)\in\mcE$; and $B_{k,e}=0$, otherwise.

Let the real-valued signals $v_n(t)$ and $i_n(t)$ denote the instantaneous voltage and current injection at node $n$; and $f_e(t)$ represent current flow on edge $e$. In what follows, the explicit time dependence of signals will be omitted for notational ease. For all lines $e=(m,n)$, the line currents and node voltages obey the first-order $RL$-dynamics 
\begin{equation} \label{eq:model-scalar}
    \ell_e\dot{f}_e+r_ef_e=v_m-v_n.
\end{equation}
The model in~\eqref{eq:model-scalar} requires $\ell_e\neq0$ to retain dynamics. Otherwise, the flows can be trivially expressed in terms of the node voltages. Thus, the following non-prohibitive assumption is made at the outset to facilitate exposition.
\begin{assumption}\label{as:1}
For all edges $e\in\mcE$, inductance $\ell_e>0$.
\end{assumption}

Define the vectors $\bv=\{v_n\}_{n\in\mcN}$, $\bi=\{i_n\}_{n\in\mcN}$, and $\bef=\{f_e\}_{e\in\mcE}$. Collectively, line dynamics in~\eqref{eq:model-scalar}, along with Kirchoff's current law (KCL) can be succinctly written using matrix-vector notation as
\begin{subequations}\label{eq:model}
	\begin{align}
	\bL\dot{\bef}&=-\bR\bef+\bB^\top\bv,\label{seq:model:RL}\\
	\bi&=\bB\bef\label{seq:model:KCL},
	\end{align}
\end{subequations}
where $\bR=\diag(\{r_e\}_{e\in\mcE})$ and $\bL=\diag(\{\ell_e\}_{e\in\mcE})$. This paper focuses on the dynamical system in~\eqref{eq:model}, identifying: input $\bv$, state $\bef$, and output $\bi$. (While this work considers a voltage-actuated network with currents serving as outputs, the developed approach can be extended to settings with current actuation and voltage outputs.) Since $\bL$ is invertible per Assumption~\ref{as:1}, the model~\eqref{eq:model} constitutes an ordinary differential equation (ODE) with a linear output equation. 

\subsection{Problem Statement}\label{sec:PS}
Suppose the network graph $\mcG$ has $N_0\geq1$ interior nodes collected in the set $\mcN_0\subset\mcN$. Without loss of generality, the network nodes can be numbered to feature the boundary nodes first, thereby enabling the partitioning $\bi=[\bi_1^\top~\bi_0^\top]^\top$ and $\bv=[\bv_1^\top~\bv_0^\top]^\top$. Corresponding to the interior nodes $\mcN_0$, we have $\bi_0=\bzero$. To explicitly impose zero-current injection for the nodes in $\mcN_0$, let us partition the incidence matrix as $\bB^\top=[\bB^\top_1~\bB^\top_0]$, where $\bB^\top_0$ has $N_0$ columns. The ensuing dynamical system is now governed by the DAE 
\begin{subequations}\label{eq:DAE}
\begin{align}
	\bL\dot{\bef}&=-\bR\bef+\bB^\top\bv,\label{eq:DAE-D}\\
	\bzero&=\bB_0\bef,\label{eq:DAE-A}
\end{align}
\end{subequations}
with the corresponding output equation
\begin{equation}
    \bi_1=\bB_1\bef\label{eq:DAEoutput}.
\end{equation}
In a nutshell, Kron reduction aspires to uncover the link between current injections $\bi_1$ and voltages $\bv_1$. Our effort is to do so by reducing the DAE \eqref{eq:DAE} to an ODE with inputs being exclusively the voltages $\bv_1$. Before we present this result, we overview prior efforts.

\section{Prior Model-reduction Results}\label{sec:prior}
\subsection{Kron Reduction in Steady State}
Consider a steady-state operating condition, wherein flows, injections, and nodal voltages are in sinusoidal steady state with frequency $\omega$. In line with the discussion in Section~\ref{prelim} for steady-state analysis, let us denote the complex-valued rotating vectors for current flow, injection, and nodal voltages by $\overrightarrow{\bef}, \overrightarrow{\bi},\overrightarrow{\bv}$, and recognize that they satisfy~\eqref{eq:model}. Substituting $\overrightarrow{\bef} = \overline{\bef} \mathrm{e}^{\jmath \omega t}$, $\overrightarrow{\bi} = \overline{\bi} \mathrm{e}^{\jmath \omega t}$, and $\overrightarrow{\bv} = \overline{\bv} \mathrm{e}^{\jmath \omega t}$ yields: 
\begin{subequations}\label{eq:temp}
	\begin{align}
	\jmath\omega\bL\overline{\bef}\mathrm{e}^{\jmath \omega t}&=-\bR\overline{\bef}\mathrm{e}^{\jmath \omega t}+\bB^\top\overline{\bv}\mathrm{e}^{\jmath \omega t},\label{seq:temp1}\\
	\overline{\bi}\mathrm{e}^{\jmath \omega t}&=\bB\overline{\bef}\mathrm{e}^{\jmath \omega t}\label{seq:temp2}.
	\end{align}
\end{subequations}
Solving for $\overline{\bef}$ from~\eqref{seq:temp1} and substituting the resultant in~\eqref{seq:temp2} yields the familiar algebraic network model
\begin{equation}\label{eq:iYv}
    \overline{\bi}=\bB(\bR+\jmath\omega \bL)^{-1}\bB^\top \overline{\bv} = \bY \overline{\bv},
\end{equation}
where $\bY=\bB(\bR+\jmath\omega \bL)^{-1}\bB^\top$ is the \emph{admittance matrix}. The inverse $(\bR+\jmath\omega \bL)^{-1}$ exists owing to the invertibility of $\bL$ per Assumption~\ref{as:1} and non-negativity of resistances $r_e\geq0,\,\forall e \in \mcE$. Suitably partitioning~\eqref{eq:iYv} provides
\begin{equation}\label{eq:Kronpartition}
    \begin{bmatrix}\overline{\bi}_1\\\mathbf{0}
     \end{bmatrix}=\begin{bmatrix}\bY_{11}&\bY_{10}\\\bY_{10}^\top&\bY_{00}\end{bmatrix}\begin{bmatrix}\overline{\bv}_1\\\overline{\bv}_0\end{bmatrix}.
\end{equation}
From the second row in~\eqref{eq:Kronpartition}, we can isolate 
$\overline{\bv}_0=-\bY_{00}^{-1}\bY_{10}^\top\overline{\bv}_1$, which, when substituted back in the first row, yields the reduced model
\begin{equation}\label{eq:Kron}
    \overline{\bi}_1=(\bY\setminus\bY_{00})\overline{\bv}_1=\bY_{r}\overline{\bv}_1,
\end{equation}
where $\bY_{r}=(\bY\setminus\bY_{00})=\bY_{11}-\bY_{10}\bY_{00}^{-1}\bY_{10}^\top$ is the Schur complement of $\bY_{00}$ of the admittance matrix, $\bY$, and is commonly referred to as the \emph{Kron-reduced admittance matrix}. The Kron-reduced admittance matrix $\bY_r$ corresponds to an equivalent connected network of series impedances~\cite{Florian-2013-Kron}. The previous manipulation relies on the invertibility of $\bY_{00}$, which is guaranteed per the result below. 
\begin{lemma}\label{lem:Y00}
Given a strict subset $\mcN_0\subset\mcN$, the submatrix $\bY_{00}$ defined as per~\eqref{eq:Kronpartition} is invertible if one of the following conditions hold: c1) $r_e>0,~\forall e\in\mcE$; c2) $l_e>0,~\forall e\in\mcE$.
\end{lemma}
\noindent A proof is provided in Appendix~\ref{app:leY00}. The sufficient condition \emph{c2)} coincides with Assumption~\ref{as:1}, thus ensuring applicability of Lemma~\ref{lem:Y00} to the networks considered in this work.
\begin{remark}
Under varying network models, results related to Lemma~\ref{lem:Y00} may be found in the recent works~\cite{Kettner18Y,turizo2020invertibility}. These  establish invertibility for principal submatrices of $\bY$ under a set of conditions including \emph{c1)}. The furnished approaches are complicated by the presence of shunt elements and transformers, see~\cite{turizo2020invertibility}. However, for the $RL$ network considered here, the relatively simpler proof for Lemma~\ref{lem:Y00} suffices.\end{remark}
\subsection{Time-domain Reduction for Homogeneous Networks}\label{sec:homo}
To reduce the time-domain model in~\eqref{eq:DAE} for eliminating the interior nodes in $\mcN_0$, previous results rely on the so called homogeneous network assumption, see~\cite{Caliskan14Automatica}. This assumption dictates that all edges, $e \in \mcE$, have a constant $r_e/\ell_e$ ratio, translating to $\bR=\alpha\bL$ for some $\alpha>0$. Substituting this homogeneity condition in~\eqref{eq:DAE-D} yields
\begin{equation}\label{eq:homoRL}
    \dot{\bef}=-\alpha\bef+\bL^{-1}\bB^\top\bv.
\end{equation}
Pre-multiplying \eqref{eq:homoRL} with $\bB$ and using KCL (as transcribed in~\eqref{seq:model:KCL}) provides the dynamic model in current injections:
\begin{equation}\label{eq:homoi}
    \dot{\bi}=-\alpha\bi+\bB\bL^{-1}\bB^\top\bv=\widetilde{\bL}\bv,
\end{equation}
where $\widetilde{\bL}=\bB\bL^{-1}\bB^\top$. Similar to $\bY$ in~\eqref{eq:iYv}, $\widetilde{\bL}$ in~\eqref{eq:homoi} corresponds to a Laplacian of the graph $\mcG$. One can partition~\eqref{eq:homoi}, use $\bi_0 =0$ (and hence $\dot{\bi}_0=0$), to obtain 
\begin{equation}\label{eq:homo-part}
    \begin{bmatrix}
             \dot{\bi}_1\\\mathbf{0}
        \end{bmatrix}=-\alpha\begin{bmatrix}
             \bi_1\\\mathbf{0}
        \end{bmatrix}+\begin{bmatrix}
                 \widetilde{\bL}_{11} &\widetilde{\bL}_{10}\\
                 \widetilde{\bL}_{10}^\top &\widetilde{\bL}_{00}
        \end{bmatrix}\begin{bmatrix}
             \bv_1\\\bv_0
        \end{bmatrix}.
\end{equation}
The structure of~\eqref{eq:homo-part} conveniently allows for elimination of the voltages $\bv_0$, as was done in transiting from \eqref{eq:Kronpartition} to \eqref{eq:Kron}. Specifically, the equation in the second row yields $\bv_0=-\widetilde{\bL}_{00}^{-1}\widetilde{\bL}_{10}^\top\bv_1$,
which, when substituted in the first row, yields the reduced dynamical model in terms of current injections:
\begin{equation}\label{eq:tdreduced}
    \dot{\bi}_1=-\alpha\bi_1+(\widetilde{\bL}\setminus \widetilde{\bL}_{00})\bv_1.
\end{equation}
The invertibility of $\widetilde{\bL}_{00}$ is ascertained from the following observation: The definition of $\widetilde{\bL}$ in \eqref{eq:homoi} indicates that for a purely inductive network with inductances given by $\bL$, the related matrix $\bY$ can be written as $\jmath\omega\widetilde{\bL}$; see~\eqref{eq:iYv}. Owing to Lemma~\ref{lem:Y00} and Assumption~\ref{as:1}, the invertibility of $\bY_{00}$ or $\jmath\omega\widetilde{\bL}_{00}$ (equivalently $\widetilde{\bL}_{00}$) is guaranteed.

The above delineated steps feature resemblance to Kron reduction in \eqref{eq:Kronpartition}-\eqref{eq:Kron}, and the reduced model~\eqref{eq:tdreduced} admits striking network-theoretic similarities as well. Specifically, matrix $(\widetilde{\bL}\setminus \widetilde{\bL}_{00})$ in~\eqref{eq:tdreduced} corresponds to a Laplacian of a reduced graph, which topologically coincides with the one obtained from~$\bY_r$ in~\eqref{eq:Kron} (see Fig.~\ref{fig:Overview}(middle)). The topology of the reduced network depends on the originating one, and is agnostic to the edge weights, see ~\cite[Prop. 5.7]{Dorfler18graph}. While not explicitly captured in the mathematical presentation, the previous approach applies to purely resistive/inductive networks as well. 

The above approach accomplishes model reduction by transforming the states from $\bef$ to $\bi$, thereby applying the zero-injection condition directly on $\bi$. However, the maneuvers involved only apply to homogeneous networks. The next section addresses a general setting. 
\section{Generalized Time-domain Model Reduction}\label{sec:general}
This section puts forth the proposed approach to eliminate zero-injection nodes $\mcN_0$ from the time-domain model~\eqref{eq:DAE}. Next, it is shown that the prior results of Section~\ref{sec:prior} can be obtained as special instances of our generalized approach. Subsequently, flexibilities in the reduced-model structure are elaborated and circuit interpretations are outlined.
\subsection{Main Result}\label{sec:main}
The $RL$ dynamics \eqref{eq:DAE-D} feature $E$ differential equations in $E$-length state-vector $\bef$. However, constraint~\eqref{eq:DAE-A} restricts the flows $\bef$ to a low-dimensional subspace; specifically $\bef\in\nullspace(\bB_0)$. It is worth noting that
$$\dimension(\nullspace(\bB_0))=E-N_0,$$
where $N_0=|\mcN_0|$. (See Appendix~\ref{app:dimension} for proof.) Therefore, one can obtain a low-dimensional embedding $\whf\in\mathds{R}^{E-N_0}$ for vectors $\bef\in\nullspace(\bB_0)$ via
\begin{equation}\label{eq:lowf}
    \bef=\bP\whf,
\end{equation}
where the matrix $\bP$ should be chosen to yield $\range(\bP)=\nullspace(\bB_0)$. In a graph-theoretic sense, matrix $\bP$ spans the space orthogonal to the cutset space defined by cuts of interior nodes in $\mcN_0$~\cite{Dorfler18graph}. Being a (potentially abstract) representation of network current flows, we refer to $\whf$ as \emph{pseudoflows}. Using the prescribed embedding, the ensuing result formally establishes the sought reduced model corresponding to~\eqref{eq:DAE}.
\begin{theorem}\label{th:TDK}
Consider the differential equation
\begin{equation}\label{eq:reducedTD}
        \whL\dot{\whf}=-\whR\whf+\whB^\top\bv_1,
\end{equation}
where,  $\whL =\bP^\top\bL\bP$, $\whR =\bP^\top\bR\bP$, and $\whB =\bB_1\bP$. The following hold:
\begin{itemize}
    \item Matrix $\whL$ is invertible rendering~\eqref{eq:reducedTD} an ODE.
    \item Solutions of network flows obtained from~\eqref{eq:reducedTD} with initial condition $\whf_0$ alongside \eqref{eq:lowf} coincide with the solutions of \eqref{eq:DAE} with $\bef(t=0)=\bP\whf_0$.
\end{itemize}
\end{theorem}
\begin{proof}
  We will start with establishing the invertibility of $\whL$. Note that positive $\ell_e$'s from Assumption~\ref{as:1} imply $\bL\succ0$. Further, since columns of $\bP$ are linearly independent, we get $\whL\succ 0$. This guarantees the invertibility of $\whL$, and renders~\eqref{eq:reducedTD} an ODE. Towards establishing the equivalence of \eqref{eq:reducedTD} and \eqref{eq:DAE}, we substitute \eqref{eq:lowf} in~\eqref{eq:DAE-D} to obtain
\begin{equation}\label{eq:LP}
    \bL\bP\dot{\whf}=-\bR\bP\whf+\bB^\top\bv.
\end{equation}
Equation~\eqref{eq:LP} constitutes an over-determined system of $E$ differential equations with linear dependence. To eliminate the linear dependence, pre-multiply~\eqref{eq:LP} with $\bP^\top$ to obtain
\begin{align}
    \bP^\top\bL\bP\dot{\whf}&=-\bP^\top\bR\bP\whf+\bP^\top[\bB_0^\top~\bB_1^\top]\begin{bmatrix}
             \bv_0\\ \bv_1
    \end{bmatrix}\notag\\
    &=-\bP^\top\bR\bP\whf+\bP^\top\bB_1^\top\bv_1,\label{eq:PLP}
\end{align}
where, the second line follows from the fact that $\range(\bP)=\nullspace(\bB_0)$, or $\bB_0\bP=\bzero$. Substituting the definitions of $(\whL,\whR,\whB)$ yields~\eqref{eq:reducedTD}.  
\end{proof}
In line with the problem statement in Section~\ref{sec:PS}, the reduced ODE model~\eqref{eq:reducedTD} eliminates the unknown voltages $\bv_0$ and features exclusively the voltages $\bv_1$ as inputs. Finally, the output equation~\eqref{eq:DAEoutput} gets modified using~\eqref{eq:lowf} to
\begin{equation}\label{eq:ODEoutput}
    \bi_1=\bB_1\bP\whf,
\end{equation}
thus yielding the sought relation from input $\bv_1$ to output $\bi_1$.
\begin{remark}[Is $\bP$ unique?]\label{rem:P}
Given matrix $\bB_0^\top$, the matrix $\bP$ featuring in the reduced model of Theorem~\ref{th:TDK} is \underline{not} unique. It can be built by collecting as columns, an arbitrary basis for $\nullspace(\bB_0)$. Technical details on choice of a specific $\bP$ and related interpretations are provided in Section~\ref{sec:P}.  
\end{remark}
\begin{remark}[Relation to Galerkin Projection]
The proposed approach features similarities to the Galerkin projection-based model order reduction (PMOR)~\cite{antoulas2020interpolatory}. In applying Galerkin PMOR to get a reduced model of order $M<N$ for \eqref{eq:DAE}, one seeks a matrix $\bV\in\mathds{R}^{E\times M}$, such that $\bef\approx\bV\whf$, where $\whf\in\mathds{R}^{M}$. Substituting this subspace approximation in~\eqref{eq:DAE-D}  yields $\bL\bV\dot{\whf}=-\bR\bV\whf+\bB^\top\bv+\br_1$, with residual $\br_1$  capturing the approximation error. Next, one pre-multiplies the previous equation by $\bV^\top$ while imposing $\bV^\top\br_1=\bzero$ towards reducing the error; thus obtaining a reduced $M$-order model. The accuracy loss from reduction is quantified by $\|\br_1\|$. While the approach of Theorem~\ref{th:TDK} is similar in spirit to Galerkin PMOR, unlike the later, it attains an equivalent reduced model with no approximation error. Specifically, with the choice of reduced-model order $M=N-E_0$, we make the subspace approximation exact in~\eqref{eq:lowf} with $\bV=\bP$; which on substitution in \eqref{eq:DAE} yields $\|\br_1\|=0$.
\end{remark}
\subsection{Prior Results as Special Cases}\label{sec:harmonizing}
This section reconciles the prior results~\eqref{eq:Kron} and~\eqref{eq:tdreduced} with the proposed generalized reduced model~\eqref{eq:reducedTD}. To this end, we first evaluate the reduced models yielded by Theorem~\ref{th:TDK} for the two special cases of Section~\ref{sec:prior}. Next we will show that these models coincide with~\eqref{eq:Kron} and~\eqref{eq:tdreduced}.
\subsubsection{Steady-state model} Assigning the steady-state form to pseudoflows as $\overline{\widehat{\mathbf{f}}}\mathrm{e}^{\jmath\omega t}$ and substituting in~\eqref{eq:reducedTD} yields (after some elementary algebraic manipulations) $(\whR+\jmath\omega\whL)\overline{\widehat{\mathbf{f}}}=\whB\bar{\bv}_1$.
Subsequently using $\overline{\bef}=\bP\overline{\widehat{\mathbf{f}}}$ and $\overline{\bi}_1=\bB_1\overline{\bef}$, from \eqref{eq:DAEoutput} and \eqref{eq:lowf} applied to steady-state values, one gets
\begin{equation}\label{eq:PsteadyKron}
    \overline{\bi}_1=\bB_1\bP\big(\bP^\top(\bR+\jmath\omega\bL)\bP\big)^{-1}\bP^\top\bB_1^\top\overline{\bv}_1,
\end{equation}
where the definitions of $(\whL,\whR,\whB)$ follow from Theorem~\ref{th:TDK}.
\subsubsection{Homogeneous networks} Substituting the homogeneous-network condition $\bR=\alpha\bL$, or equivalently $\whR=\alpha\whL$, in~\eqref{eq:reducedTD} and premultipying by $\whL^{-1}$ provides
$$\dot{\whf}=-\alpha\whf+\whL^{-1}\whB^\top\bv_1.$$
Further, premultiplying by $\bB_1\bP$ (invoking~\eqref{eq:DAEoutput} and \eqref{eq:lowf}), and substituting the definitions of $(\whL,\whB)$ yields
\begin{equation}\label{eq:PhomoKron}
    \dot{\bi}_1=-\alpha\bi_1+\bB_1\bP(\bP^\top\bL\bP)^{-1}\bP^\top\bB_1^\top\bv_1.
\end{equation}

Next, we show the equivalence of the reduced models \eqref{eq:PsteadyKron} and \eqref{eq:PhomoKron} to prior results~\eqref{eq:Kron} and \eqref{eq:tdreduced} in two steps. First, we note that \eqref{eq:PsteadyKron} and \eqref{eq:PhomoKron} feature matrix $\bP$, which is not unique.  Hence, we provide an algebraic claim in Lemma~\ref{le:PWP} to  eliminate the apparent ambiguity from the non-uniqueness of $\bP$. Second, we show in Proposition~\ref{prop:unify} that via appropriate instantiation of a weighting matrix, the models in~\eqref{eq:PsteadyKron} and~\eqref{eq:PhomoKron} coincide with the prior results~\eqref{eq:Kron} and~\eqref{eq:tdreduced}. The proof of Proposition~\ref{prop:unify} builds upon the technical Lemma~\ref{le:PWP}.
\begin{lemma}\label{le:PWP}
For a connected graph $(\mcN,\mcE)$, consider complex-valued edge weights $w_e\neq 0, \forall e\in\mcE$, and a row-block partition of the companion incidence matrix as $\bB^\top=[\bB^\top_1~\bB^\top_0]$; define $\bW = \diag(\{w_e\}_{e\in\mcE})$. Given a full-column-rank matrix $\bP$ with $\range(\bP)=\nullspace(\bB_0)$, it holds that
           \begin{align}\label{eq:LePWP}
     \begin{split}
         \bP(\bP^\top\bW\bP)^{-1}\bP^\top&=\\\bW^{-1}-\bW^{-1}&\bB_0^\top(\bB_0\bW^{-1}\bB_0^\top)^{-1}\bB_0\bW^{-1}.
                          \end{split}
                          \end{align}
\end{lemma}
\begin{proof}
  Define $\wtB_0=\bB_0\bW^{-\frac{1}{2}}$ and $\wtP=\bW^{\frac{1}{2}}\bP$. Given ${\range(\bP)=\nullspace(\bB_0)}$, it follows that $\range(\wtP)=\nullspace(\wtB_0)$ and the LHS of~\eqref{eq:LePWP} can be written as
  \begin{equation}\label{eq:ProjP}
      \bP(\bP^\top\bW\bP)^{-1}\bP^\top=\bW^{-\frac{1}{2}}\underbrace{\wtP(\wtP^\top\wtP)^{-1}\wtP^\top}_{= \mathds{P}_{\wtP}}\bW^{-\frac{1}{2}},
  \end{equation}
  where $\mathds{P}_{\wtP}$ denotes the projection onto $\range(\wtP)$ (in this case, equivalent to the projection onto $\nullspace(\wtB_0)$). Denote the projection onto $\range(\wtB_0^\top)$ as $\mathds{P}_{\wtB_0^\top} =\wtB_0^\top(\wtB_0\wtB_0^\top)^{-1}\wtB_0$, where the invertibility of $\wtB_0\wtB_0^\top=\bB_0\bW^{-1}\bB_0^\top$ is guaranteed since $\bB_0$ is a strict row partition of the incidence matrix of a connected graph, and hence features linearly independent rows. Using $\nullspace(\wtB_0)=\range(\wtB_0^\top)$, it can be shown that
  \begin{equation}\label{eq:sumI}
     \mathds{P}_{\wtP}=\mathds{I}- \mathds{P}_{\wtB_0^\top},
  \end{equation}
  where $\mathds{I}$ is the identity matrix of size $E$. Pre- and post-multiplying \eqref{eq:sumI} with $\bW^{-\frac{1}{2}}$ and substituting the definitions of $(\mathds{P}_{\wtP}, \mathds{P}_{\wtB_0^\top})$ provides
       \begin{align}
       \begin{split}
         \bW^{-\frac{1}{2}}\wtP(\wtP^\top\wtP)^{-1}\wtP^\top\bW^{-\frac{1}{2}}&=\\\bW^{-1}-\bW^{-\frac{1}{2}}&\wtB_0^\top(\wtB_0\wtB_0^\top)^{-1}\wtB_0\bW^{-\frac{1}{2}}.
         \end{split}
     \end{align}
     Substituting $\wtB_0=\bB_0\bW^{-\frac{1}{2}}$ and $\wtP=\bW^{\frac{1}{2}}\bP$ above yields~\eqref{eq:LePWP} and completes the proof.
\end{proof}

Lemma~\ref{le:PWP} establishes that despite the non-uniqueness of $\bP$, the structural form in the LHS of~\eqref{eq:LePWP} equals the unique matrix in the RHS. The next result, proved in Appendix~\ref{app:Prop}, uses Lemma~\ref{le:PWP} to reconcile \eqref{eq:PsteadyKron} to \eqref{eq:Kron}, and \eqref{eq:PhomoKron} to \eqref{eq:tdreduced}.
\begin{proposition}\label{prop:unify}
For the complex edge weights $\{w_e\}$ introduced in Lemma~\ref{le:PWP}, define Laplacian matrix $\widetilde{\bW}=\bB\bW^{-1}\bB^\top$, and let $\widetilde{\bW}_{00}$ be the matrix block partition corresponding to the interior nodes $\mcN_0$. Then, it holds that
\begin{equation}\label{eq:cor}
    \bB_1\bP(\bP^\top\bW\bP)^{-1}\bP^\top\bB_1^\top=\widetilde{\bW}\setminus\widetilde{\bW}_{00}.
\end{equation}
Choosing $\bW=\bR+\jmath\omega\bL$ in~\eqref{eq:cor} establishes the equivalence of~\eqref{eq:PsteadyKron} and~\eqref{eq:Kron}; and choosing $\bW=\bL$ establishes the equivalence of~\eqref{eq:PhomoKron} and~\eqref{eq:tdreduced}.
\end{proposition}
\subsection{Choices of $\bP$ and Related Interpretations}\label{sec:P}
\begin{enumerate}
    \item Given the flows on all-but-one edges incident on an interior node, one can trivially recover the extant flow via KCL. Thus, a sub-vector $\whf$ built by omitting from $\bef$, one incident-edge flow per interior node can serve as a reduced representation. Hence, one can express $\bef=\bP\whf$ by stacking the rows of $\bP$ as suitable canonical vectors corresponding to the retained flows; and as $\{0,\pm 1\}$ vectors to calculate the omitted flows.
    \item Compute the basis vectors spanning $\nullspace{(\bB_0)}$. Stack these as columns to obtain $\bP$. 
    \item One can choose $\bP$ to yield diagonal $(\whL,~\whR)$ via the following steps: i) Compute a matrix $\bP'$ to span $\nullspace{(\bB_0^\top)}$ as described in choice 2); ii) Set $\whL'=\bP'^\top\bL\bP'$ and $\whR'=\bP'^\top\bR\bP'$; iii) Obtain the generalized eigen-value decomposition $\whL'\bV=\whR'\bV\bD$; and iv) set $\bP=\bP'\bV$. It can be verified that thus-obtained $\bP$ yields diagonal $(\whL,~\whR)$ with non-negative entries. The latter stems from using $\bL,\bR\succeq0$ in the definitions of $\whL,\whR$.  The reduced model~\eqref{eq:reducedTD} can then be interpreted as $(E-N_0)$ independent $RL$-circuits actuated by a linear combination of voltages in $\bv_1$; see Fig.~\ref{fig:Overview}(bottom). The individual equations read as $$\widehat{L}_{kk}\dot{\hat{f}}_k +\widehat{R}_{kk}\hat{f}_k=\sum_{n=1}^{N-N_0} \beta_{kn}v_n,$$
    where $\beta_{kn}$'s are entries of the product $\bP^\top\bB_1^\top$.
\end{enumerate}
\section{Numerical Tests}\label{tests}
This section empirically illustrates the effectiveness of the proposed generalization using $\wye$-$\Delta$ transformation as an example. For comparison, we adopt a baseline approach from ~\cite{Floriduz19Kron} that heuristically extends Kron reduction to the time domain. Denoted as~$(\mathcal{B})$, the approach involves the following steps: $(\mathrm{S}1)$~Given a $\wye$-connected $RL$ network with one interior node (see Fig.~\ref{fig:Overview}), choose a frequency $\omega_0$ and build an admittance matrix per~\eqref{eq:iYv} that corresponds to a $\wye$-connected \emph{impedance} network. $(\mathrm{S}2)$~Obtain a reduced $\Delta$-connected impedance network using~\eqref{eq:Kron}. $(\mathrm{S}3)$~Separate out real and imaginary parts of impedances in the reduced $\Delta$ connection and factor out~$\omega_0$ to recover a $\Delta$-connected $RL$ network. Clearly, $(\mathcal{B})$ lacks a theoretical claim of equivalence in the time domain; it is further accompanied by two implementation ambiguities: $(\mathrm{A}1)$~To obtain a time-domain solution of the $\wye$-connected $RL$ network, one would be presented with the initial conditions $\bef^{\wye}_0\in\mathds{R}^{3\times 1}$. On obtaining the $\Delta$-connected $RL$ network in $(\mathrm{S}3)$, how does one obtain a corresponding $\bef^{\Delta}_0$? $(\mathrm{A}2)$~How to choose $\omega_0$ in $(\mathrm{S}1)$ when the network may be actuated by arbitrary voltages? 

\subsection{Setup}
Two sets of numerical tests will be presented next to compare our approach, $(\mathcal{B})$, and a ground-truth DAE model implemented in the Matlab-Simulink environment. Both are conducted for a $\wye$-connected $RL$ network with $RL$ parameters randomly drawn from $[0.5,1]$, and given by $\bR=\diag([0.98~0.99~0.58])~\Omega$ and $\bL=\diag([0.55~0.64~0.77])~\mathrm{H}$; and initial branch currents $\bef^{\wye}_0=-[5,5,-10]^\top~\mathrm{A}$. To obtain an initial condition for the current flows $\bef^\Delta_0$ given $\bef^{\wye}_0$ in $(\mathcal{B})$, we have to ensure the boundary-node current injections agree in both $\Delta$ and $\wye$ representations. One such solution, $\tilde{\bef}^\Delta_0$ can be obtained using Matlab command~\texttt{lsqminnorm}$(\bB_r,\bef^{\wye}_0)$, where $\bB_r$ is the incidence matrix of the $\Delta$ connection. This operation returns a minimum-norm solution to $\bB_r \bef^\Delta_0 = \bef^{\wye}_0$. With cyclic edge direction assignment in the $\Delta$ connection, we have $\bB_r\bone=\bzero$; hence for any scalar $\gamma$, vector $\tilde{\bef}^\Delta_0+\gamma\bone$ conforms with the required initial-current injections. Given the ambiguity, $5$ random instances of $\gamma$ are drawn uniformly from $[-5,5]~\mathrm{A}$ in our execution of~$(\mathcal{B})$.
\begin{figure}[t!]
\centering
\includegraphics[scale=0.45]{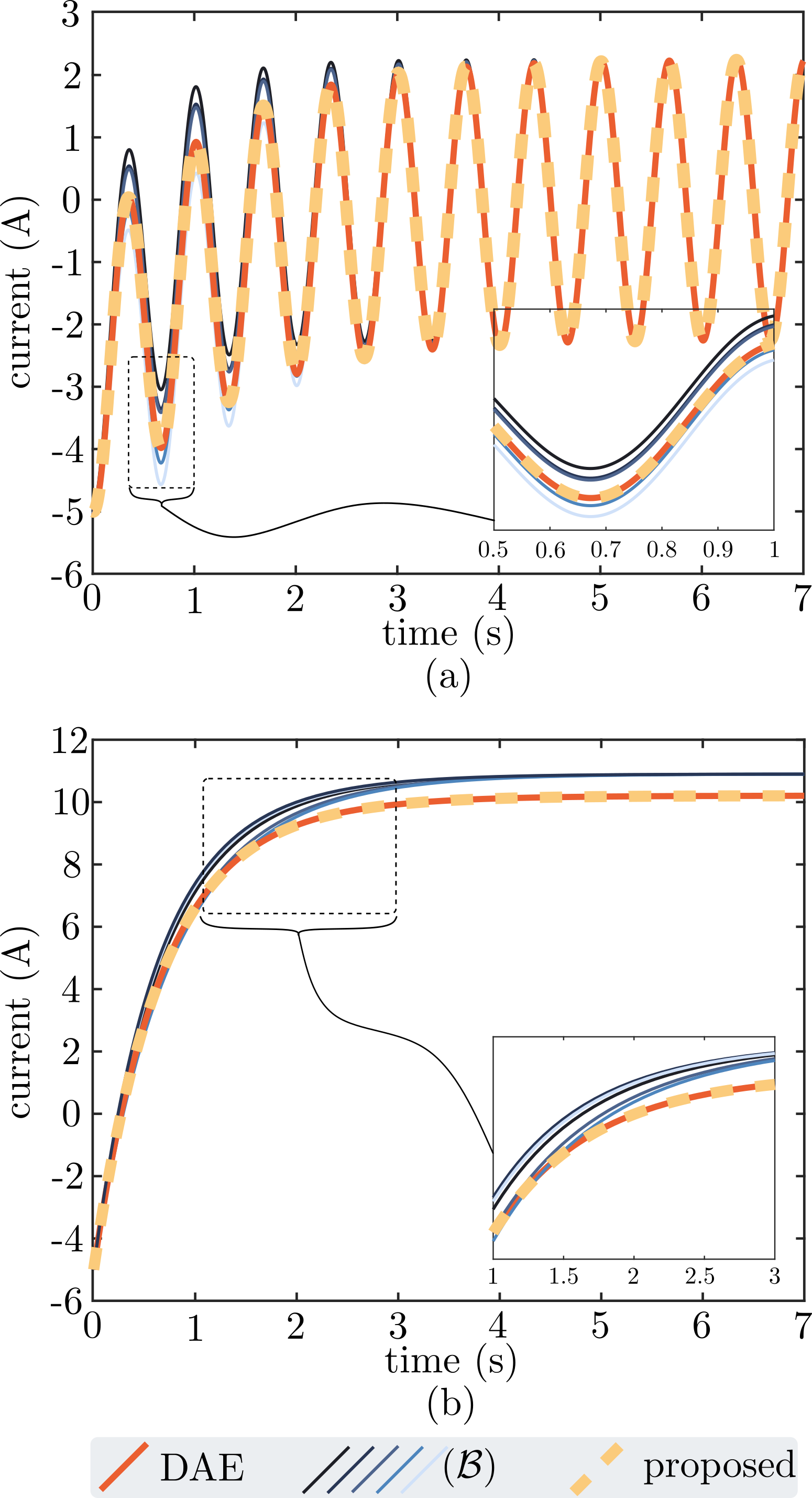}
\vspace{-1em}
\caption{Current injections $i_1(t)$ obtained with: a) sinusoidal-voltage excitation; and b) step-voltage excitation.}
\label{fig:sinusoid}
\vspace{-1em}
\end{figure}

\subsection{Results \& Inferences}
In the first set of tests, we apply $1.5~\mathrm{Hz}$ sinusoidal voltages with amplitudes $120~\mathrm{V}$ and phases $(0,30^\circ,-30^\circ)$ to the three boundary nodes. To implement our approach, $\bP$ was obtained using Matlab command \texttt{null}$(\bA_0)$, the initial condition for simulating the proposed reduced model~\eqref{eq:reducedTD} was evaluated from~\eqref{eq:lowf}. To implement~$(\mathcal{B})$, $\omega_0$ was picked to be $2\pi\times1.5~\mathrm{rad}\cdot\mathrm{s}^{-1}$. The current injections at node $1$, $i_1(t)$, obtained via the three approaches are illustrated in Fig.~\ref{fig:sinusoid}(a). The following observations are in order: i)~the initial values for all approaches coincide by design; ii)~results from the reduced model obtained with the proposed approach coincide point wise with the ground-truth DAE model; and iii)~results from $(\mathcal{B})$ with randomized initializations vary during transient conditions but align with the ground-truth DAE model in steady state~(aligning with the claims in~\cite{Floriduz19Kron}).

In the next set of tests, the excitation voltages were changed to be step functions assuming steady-state values $[120, 100, 110]^\top~\mathrm{V}$. The current injections at node $1$, $i_1(t)$, obtained via the three approaches are illustrated in Fig.~\ref{fig:sinusoid}(b). One observes that: i)~results from the reduced model obtained with the proposed approach coincide point wise with the ground-truth DAE model; and ii)~results of $(\mathcal{B})$ with randomized initializations coincide with each other in steady state, \emph{but} all are different from the ground-truth. 

These sets of numerical tests demonstrate that while the proposed model reduction holds for arbitrary voltage-actuated $RL$ networks in the time domain, an attempt to heuristically extend the classical Kron reduction may not yield accurate results (even in steady state).
\section{Concluding Remarks}
This work put forth a time-domain generalization of Kron reduction for $RL$ networks. Prominent prior results for steady-state conditions and homogeneous networks were shown to be special instances of the proposed model. Numerical tests on the well-known $\wye-\Delta$ transformation setup validated the approach and highlighted the limitations of existing heuristics. Given that the projection matrix $\bP$ relates to the cut-set space of the underlying graph, it is tempting to further investigate graph-theoretic interpretations of the proposed low-dimensional embedding. 
\appendix
\subsection{Proof of Lemma~\ref{lem:Y00}}\label{app:leY00}
Having $r_e,~\ell_e\geq0$ entails $\Re(\bY)$ and $\Im(\bY)$ are positive semidefinite; implying $\Re(\bY_{00}),~\Im(\bY_{00})\succeq 0$. Next, we prove the invertibility of $\bY_{00}$ assuming condition \emph{c1)} is satisfied. The proof for \emph{c2)} follows similarly. Given $r_e>0~\forall e$, the matrix $\Re(\bY)$ is a Laplacian matrix. Furthermore, matrix $\Re(\bY_{00})$ is a strict principal submatrix of $\Re(\bY)$ as $\mcN_0\subset\mcN$; hence $\Re(\bY_{00})\succ 0$. Proving by contradiction, let us assume that the matrix $\bY_{00}$ is singular, implying $\exists~\bx=\bx_r+\jmath \bx_i\neq\bzero$ such that $\bY_{00}\bx=\bzero$. Separating the real and imaginary parts of $\bY_{00}\bx=\bzero$ reads
\begin{subequations}
	\begin{align}
	\Re(\bY_{00})\bx_r-\Im(\bY_{00})\bx_i=\bzero,\label{seq:con1}\\
	\Im(\bY_{00})\bx_r+\Re(\bY_{00})\bx_i=\bzero\label{seq:con2}.
	\end{align}
\end{subequations}
Since $\Re(\bY_{00})\succ 0$,~\eqref{seq:con1} yields \begin{equation}
    \bx_r=[\Re(\bY_{00})]^{-1}\Im(\bY_{00})\bx_i\label{eq:xr}.
\end{equation}
Substituting $\bx_r$ from \eqref{eq:xr} in \eqref{seq:con2} yields \begin{equation}
    \left(\Im(\bY_{00})[\Re(\bY_{00})]^{-1}\Im(\bY_{00})+\Re(\bY_{00})\right)\bx_i=\bzero.\label{eq:xi}
\end{equation}
Using $\Im(\bY_{00})\succeq 0$, $\Re(\bY_{00})\succ 0$, one finds $\left(\Im(\bY_{00})[\Re(\bY_{00})]^{-1}\Im(\bY_{00})+\Re(\bY_{00})\right)$ is invertible, implying $\bx_i=\bzero$ from~\eqref{eq:xi}. Further,~\eqref{eq:xr} yields $\bx_r=\bzero$, or $\bx=\bzero$, leading to a contradiction; thus, $\bY_{00}$ is invertible.
\subsection{Proof for $\dimension(\nullspace(\bB_0))=E-N_0$}\label{app:dimension}
Since $\bB$ is the incidence matrix of a \emph{connected} graph, it features $N-1$ linearly independent rows with $\bone^\top\bB=\bzero$. Thus, any selection of $N_0<N$ rows are linearly independent. Hence, the rows of matrix $\bB_0$ are linearly independent implying the dimension of $\nullspace(\bB_0)$ is $E-N_0$.
\subsection{Proof of Proposition~\ref{prop:unify}}\label{app:Prop}
From the definition $\widetilde{\bW} =\bB\bW^{-1}\bB^\top$, and the partition $\bB^\top=[\bB^\top_1~\bB^\top_0]$, matrix $\widetilde{\bW}$ can be expressed as
  \begin{align*}\label{eq:propProofW}
      \widetilde{\bW}=\begin{bmatrix}\widetilde{\bW}_{11}&\widetilde{\bW}_{10}\\\widetilde{\bW}_{10}^\top&\widetilde{\bW}_{00}\end{bmatrix}&=      \begin{bmatrix}
               \bB_1\\ \bB_0
      \end{bmatrix}\bW^{-1}[
               \bB_1^\top~~\bB_0^\top]\nonumber \\
               &=\begin{bmatrix}
               \bB_1\bW^{-1}\bB_1^\top & \bB_1\bW^{-1}\bB_0^\top\\
               \bB_0\bW^{-1}\bB_1^\top & \bB_0\bW^{-1}\bB_0^\top
      \end{bmatrix}.
  \end{align*}
  The Schur complement $\widetilde{\bW}\setminus\widetilde{\bW}_{00}$ is given by
  \begin{align*}
      &\widetilde{\bW}\setminus\widetilde{\bW}_{00}=\widetilde{\bW}_{11}-\widetilde{\bW}_{10}\widetilde{\bW}_{00}^{-1}\widetilde{\bW}_{10}^\top\\
      &=\bB_1\bW^{-1}\bB_1^\top-\bB_1\bW^{-1}\bB_0^\top(\bB_0\bW^{-1}\bB_0^\top)^{-1}\bB_0\bW^{-1}\bB_1^\top\\
      &=\bB_1\big(\bP(\bP^\top\bW\bP)^{-1}\bP^\top\big)\bB_1^\top,
  \end{align*}
  where, the second equality follows from substituting the matrix blocks, and the last equality follows from Lemma~\ref{le:PWP}. 

\section*{Acknowledgement}
Assistance from D.~Venkatramanan on simulation validation and graphics is greatly appreciated. 

\bibliographystyle{IEEEtran}
\bibliography{myabrv,power}

\end{document}